\documentclass[conference]{IEEEtran}

\ifCLASSINFOpdf
\else
\fi

\usepackage{pgfplots}
\usepackage{tikz}
\pgfplotsset{height=0.2\textheight, width=0.4\textwidth}

\usepackage{graphicx}
\usepackage{mathtools}
\usepackage{amsmath}
\usepackage{amsfonts}
\usepackage{enumerate}
\usepackage{dsfont}
\usepackage{amssymb}
\usepackage{bbold}

\newtheorem{theorem}{Theorem}[section]
\newtheorem{lemma}[theorem]{Lemma}
\newtheorem{proposition}[theorem]{Proposition}
\newtheorem{corollary}[theorem]{Corollary}

\newenvironment{proof}[1][Proof]{\begin{trivlist}
\item[\hskip \labelsep {\bfseries #1}]}{\end{trivlist}}

\newenvironment{remark}[1][Remark]{\begin{trivlist}
\item[\hskip \labelsep {\bfseries #1}]}{\end{trivlist}}

\newcommand{\qed}{\nobreak \ifvmode \relax \else
      \ifdim\lastskip<1.5em \hskip-\lastskip
      \hskip1.5em plus0em minus0.5em \fi \nobreak
      \vrule height0.75em width0.5em depth0.25em\fi}

\begin{document}
%
\title{Displacement Convexity -- A Useful Framework for the Study of Spatially Coupled Codes}

\author{\IEEEauthorblockN{Rafah El-Khatib}
\IEEEauthorblockA{EPFL, Switzerland\\
Email: rafah.el-khatib@epfl.ch}
\and
\IEEEauthorblockN{Nicolas Macris}
\IEEEauthorblockA{EPFL, Switzerland\\
Email: nicolas.macris@epfl.ch}
\and
\IEEEauthorblockN{Ruediger Urbanke}
\IEEEauthorblockA{EPFL, Switzerland\\
Email: ruediger.urbanke@epfl.ch}}

\maketitle

\begin{abstract}
Spatial coupling has recently emerged as a powerful paradigm to
construct graphical models that work well under low-complexity
message-passing algorithms. Although much progress has
been made on the analysis of spatially coupled models under message
passing, there is still room for improvement, both in terms of
simplifying existing proofs as well as in terms of proving
additional properties.

We introduce one further tool for the analysis, namely the concept
of displacement convexity. This concept plays a crucial role in the
theory of optimal transport and, quite remarkably, it is also well suited for the
analysis of spatially coupled systems. In cases where the concept applies,
displacement convexity allows functionals of distributions which
are not convex in the usual sense to be represented in an alternative form, so that
they are convex with respect to the new parametrization.
As a proof of concept we consider spatially coupled $(l,r)$-regular Gallager ensembles
when transmission takes place over the binary erasure channel. We show that the potential function of the coupled
system is displacement convex. Due to possible translational degrees
of freedom  convexity by itself falls short of establishing the
uniqueness of the minimizing profile. For the spatially coupled $(l,r)$-regular system strict displacement convexity holds when a global translation
degree of freedom is removed. Implications for the uniqueness of the minimizer and for solutions of 
the density evolution equation are discussed.

\end{abstract} \IEEEpeerreviewmaketitle

\section{Introduction}
Spatially coupled codes were introduced in the form of low-density
parity-check codes by Felstrom and Zigangirov in \cite{Zigangirov}.
Such codes are constructed by spatially coupling nearby replicas
of a code defined on a graph. It has been proven that such ensembles
perform very well under low-complexity message-passing algorithms.
Indeed, this combination achieves essentially optimal performance.
More generally, the concept of spatial coupling is proving to be
very useful not only for coding but also in compressive sensing, statistical physics, and
random constraint solving problems. Given this range of applications, it is
worth investigating basic properties of this construction in
generality. Our aim is to introduce one further tool for the analysis
of such systems -- namely the concept of displacement convexity.
Displacement convexity plays a crucial role in the theory of optimal
transport. But it is also very well suited as a tool for the analysis
of spatially coupled graphical models.


One of the most important properties of spatially coupled codes is
that they exhibit the so-called threshold saturation phenomenon.
That is, spatially coupled ensembles generically have a BP threshold
which is as large as the maximum-a posteriori (MAP) threshold of the
underlying ensemble, i.e., their threshold has {\em saturated} to
the largest possible value.  This result has been proved for
transmission over the BEC in \cite{ShriniWhy}, \cite{Pfister}, and \cite{KRU12}
and for transmission over general binary-input memoryless
output-symmetric (BMS) channels in \cite{Shrinivas} and
\cite{PfisterMacris}.

The tool we introduce probably has a considerably larger range
of applications for (coding) systems which
are governed by a variational principle. We use $(l,r)$-regular
Gallager ensembles as a proof of concept. We run the belief-propagation
algorithm on this ensemble and express it in the variational form
using the potential functional \cite{Pfister}, \cite{PfisterMacris}.
The potential functional essentially is an average of the Bethe free energy on the ensemble 
and the channel output. The density evolution (DE)
equations can be obtained by differentiating this potential. In particular, the minimizers of the potential 
are given by the ``erasure probablity profiles'' that are solutions of the DE equations. 

For the $(l,r)$-regular ensemble, we use displacement
convexity \cite{McCann} to prove that the potential describing the system is
convex with respect to an alternative structure of probability
measures. Roughly speaking these are the probability measures associated to erasure probability profiles viewed as cumulative distribution functions
(cdf's). 
We consider the static case, when the decoder phase
transition threshold is equal to the maximum a posteriori (MAP)
threshold. This displacement convexity property plays a fundamental role in characterizing the set of
minimizing profiles in a suitable space of {\it increasing profiles}. For our case we get strict 
displacement convexity once a global translational degree 
of freedom of the profiles is removed. These results allow to conclude that in a suitable space of 
increasing profiles the minimizer of the potential functional is unique up to translations, and so is the solution of the DE equation. 

We also look at the minimization problem in a more general space of profiles that are not necessarily increasing. We show that 
the potential functional satisfies rearrangement inequalities which allow to reduce the search for minimizers to a space of increasing profiles.
However it is not clear if and when the rearrangement inequalities are strict. As a result the analysis falls short of establishing that there cannot
exist non increasing profiles that are minimizers and solutions of the DE equation.

The paper is organized as follows: Section \ref{introSettingSection}
introduces the framework for our analysis and our main results. We
then give a quick introduction of the notion of displacement convexity
in Section \ref{displConvDefSection}. Finally, Section
\ref{existenceSection} presents a proof of existence of the profile
that minimizes the potential, and Section \ref{displConvSection}
proves that the functional is displacement convex. We discuss possible generalizations and open issues in the conclusion.
\begin{figure}
  \centering
  \begin{tikzpicture}
    \begin{axis}[xlabel = $p$, ylabel = $W_s(p)$, grid=both, xmin=-0.1, ymin=0, xmax=1]
      \addplot[no marks, blue, domain = 0:1, smooth]
      {1.0/6.0 + (1-x)*( 5.0/6.0 * (1-x)^(0.2) - 1 ) - 0.4881 / 3.0 * x^(3.0)};
       \addplot[no marks, blue, domain = -0.1:0, smooth, dashed]
            {1.0/6.0 + (1-x)*( 5.0/6.0 * (1-x)^(0.2) - 1 ) - 0.4881 / 3.0 * x^(3.0)};

      \node [above] at (axis cs: 0.9,0) {$P_{MAP}$};
    \end{axis}
  \end{tikzpicture} \label{singlePotPlotat4881} 
  \caption{The plot of the single system potential $W_s(p)$ as a
  function of the check-erasure probability $p$, for a $(3,6)$ uncoupled
  ensemble and $\epsilon=\epsilon_{\rm MAP}$. There are two minima
  at $p=0$ and $p=p_{\rm MAP}$.}
\end{figure}
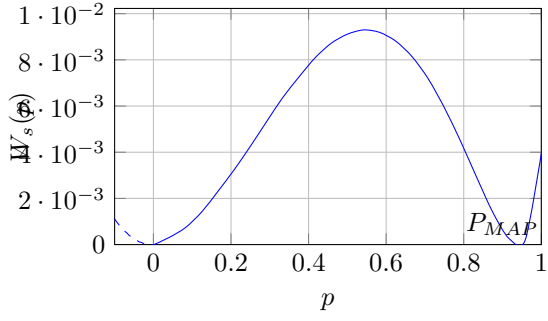

\section{Setting and Main Results} \label{introSettingSection}
In this section, we introduce the model and the associated variational
problem to which we apply the displacement convexity proof technique,
and we state our main result.

\subsection{$(l,r, L, w)$-Regular Ensembles on the BEC} \label{introBECSection}
Consider the spatially coupled $(l,r,L,w)$-regular ensemble, described
in detail in \cite{Shrinivas}, where the parameters represent
the left degree, right degree, system length, and coupling
window size (or smoothing parameter), respectively.  Specifically,
the ensemble is constructed as follows: consider $2L+1$ replicas
of a protograph of an $(l,r)$-regular ensemble. We couple these
components by connecting every variable node to $l$ check nodes,
and every check node to $r$ variable nodes.  The connections are
chosen randomly: for a variable node at position $z$, each of its
$l$ connections is chosen uniformly and independently in the range
$[z,\dots,z+w-1]$, and for a check node at position $z$, each of
its $r$ connections is chosen uniformly and independently in the
range $[z-w+1,\dots,z]$.

For the channel we take a BEC with parameter $\epsilon$.
Let $\check{x}_{z};\;z\in [-L,\dots,L]$ denote the erasure probability
of the variable node at position $z$. Consider the average over a window $x_z\equiv \frac{1}{w}\sum\limits_{i=0}^{w-1}\check{x}_{z-k}$. Then the fixed-point (FP) condition implied by density evolution (DE) is
\begin{equation*}
\displaystyle x_{z}=\frac{\epsilon}{w}\sum\limits_{k=0}^{w-1}\Big(1-\frac{1}{w}\sum\limits_{i=0}^{w-1}(1-x_{z-k+i})^{r-1}  \Big)^{l-1}.
\end{equation*}
This FP condition can be obtained by minimizing a ``potential functional", which is
\begin{equation}
\begin{split}
\frac{1}{w}\sum_{z=-L}^{L}\Bigg\{-x_{z}(1-x_{z})^{r-1}+\frac{1}{r}-\frac{1}{r}(1-x_{z})^r\\
-\frac{\epsilon}{l}\Big(\frac{1}{w}\sum_{u=0}^{w-1}(1-(1-x_{z+u})^{r-1})\Big)^{l}\Bigg\}.
\end{split} \label{discretePotential}
\end{equation}
At this point the normalization $1/w$ is a convenience whose reason will immediately appear.

The natural setting for displacement convexity is the continuum case.
We will therefore consider the continuum limit of
\eqref{discretePotential}. Extending our results to the discrete
setting is one among various open problems. We define the rescaled variables
$\tilde{z}=\frac{z}{w}$, $\tilde{u}=\frac{u}{w}$ and the rescaled function
$\tilde{x}(\frac{z}{w})\equiv x_{z}$.  It is easy to see that
\eqref{discretePotential} becomes a Riemann sum. When we take the
limit $L\to +\infty$ first and then $w\to +\infty$, we find
\begin{equation}
\begin{split}
\displaystyle\int\limits_{\mathds{R}}\,\mathrm{d}\tilde{z}\,&\Bigg\{ -\tilde{x}(\tilde{z})(1-\tilde{x}(\tilde{z}))^{r-1}-\frac{1}{r}(1-\tilde{x}(\tilde{z}))^r\\
&+\frac{1}{r}-\frac{\epsilon}{l} \Big(\int\limits_{0}^{1}\,\mathrm{d}\tilde{u}\,(1-(1-\tilde{x}(\tilde{z}+\tilde{u}))^{r-1})\Big)^{l}\Bigg\}.
\label{minFormulation1}
\end{split}
\end{equation}
At this point the reader might wonder if the integrals converge.
As explained in the introduction, we look in this paper at the decoder
phase transition threshold $\epsilon = \epsilon_{\text{\tiny MAP}}$. We give
at the end of this paragraph the conditions on the erasure probability
profile needed to have a well defined problem. From now on, the
reader should think of the noise level as fixed to the value
$\epsilon=\epsilon_{\text{\tiny MAP}}$, although we abuse notation
by simply writing $\epsilon$ in the formulas that follow.

It is more convenient to express \eqref{minFormulation1} with the function
$p(z)=1-(1-\tilde{x}(z))^{r-1}$. Note that this function is interpreted
as the erasure probability emitted by check nodes.  Summarizing,
the potential functional of interest is
\begin{equation}
\begin{split}
\displaystyle\mathcal{W}[p(\cdot)]=\int\limits_{\mathds{R}}\mathrm{d}z\Bigg\{\displaystyle\Big(&1-\frac{1}{r}\Big)\Big(1-p(z)\Big)^{\frac{r}{r-1}}-(1-p(z))\\
+&
\frac{1}{r}-\frac{\epsilon}{l}\Big(\int_{0}^{1}\mathrm{d}u\,p(z+u)\Big)^{l}\Bigg\}
\label{contPotential}.
\end{split}
\end{equation}
A word about the notation here: we use {\it square} brackets for functionals i.e., ``functions of functions'' and usual {\it round} brackets
for functions of a real variable. The continuum limit of the DE equation expressed in terms of $p(z)$ reads
\begin{align}\label{dep}
 1 - (1 -p(z))^{\frac{1}{r-1}} = \epsilon\int_{0}^1 dv \Big(\int_{0}^{1}\mathrm{d}u\,p(z+u-v)\Big)^{l-1}
\end{align}
One can check that \eqref{dep} gives the stationary points of \eqref{contPotential}.

Equation \eqref{contPotential} can be expressed as a sum of two contributions $\mathcal{W}_{\rm
single}[p(\cdot)]+\mathcal{W}_{\rm int}[p(\cdot)]$ which are defined as follows:
\begin{align}
\mathcal{W}_{\rm single}[p(\cdot)]=\displaystyle\int\limits_{\mathds{R}}&\mathrm{d}z\Bigg\{\displaystyle\Big(1-\frac{1}{r}\Big)(1-p(z))^{\frac{r}{r-1}}-(1-p(z)) \nonumber \\ &
+\frac{1}{r}-\frac{\epsilon}{l}p(z)^{l}\Bigg\}\equiv \int_{\mathds{R}}\mathrm{d}z\,W_{s}(p(z)),
\label{singleTermIntegral1}
\end{align}
\begin{equation}
\mathcal{W}_{\rm int}[p(\cdot)]=\int\limits_{\mathds{R}}
\mathrm{d}z\,\frac{\epsilon}{l}\Bigg\{p(z)^{l}-\Big(\int_{0}^{1}\mathrm{d}u\,p(z+u)\Big)^{l}\Bigg\}\label{interactionTerm1}.
\end{equation}
We call \eqref{singleTermIntegral1} the ``single system potential
functional" and \eqref{interactionTerm1} the ``interaction
functional".  The following remarks explain the interpretation
suggested by these names.  The term \eqref{interactionTerm1} vanishes
when evaluated for a constant $p(z) = p$. Moreover the integrand
of \eqref{singleTermIntegral1}, namely $W_s(p(z)) = W_s(p)$ is just
the potential of the underlying uncoupled code ensemble. This is
easily seen by recognizing that the usual DE equation
for the erasure probability of checks is recovered by setting the
derivative of $W_s(p)$ to zero. We will call $W_s(p)$ the ``single
system potential".  A plot of $W_{s}(p)$ for the $(3,6)$-ensemble
is shown in Figure \ref{singlePotPlotat4881} when
$\epsilon=\epsilon_{\text{\tiny MAP}}$. The figure shows that the single
potential vanishes at $p=0$ and  $p=p_{\text{\tiny MAP}}$, some
positive value. This is a generic feature of all $(l,r)$-regular
code ensembles as long as $l\geq 3$ (for cycle codes $l=2$, we have
$p_{\text{\tiny MAP}}=0$). This shows, in particular, that in order
for the integrals in \eqref{contPotential} to be well defined, we
have to consider profiles $p(z)$ that tend quickly enough to values, as $z\to \pm \infty$, such that the potential vanishes. Besides we 
will, for simplicity, restrict ourselves to continuous profiles. 

The above remarks motivate us to define the following spaces of profiles. Let 
\begin{align}
\mathcal{S}=\{ p(\cdot):\, & \mathds{R}\rightarrow \mathds{R}^+ \,{\rm continuous~ and~s.t}\, \nonumber \\ & \lim_{z
\rightarrow -\infty} z p(z)=0 , 
\lim_{z \rightarrow +\infty}z(p(z) - p_{\text{\tiny MAP}})=0\}.
\label{spaceS1}
\end{align}
Note that the left limit is $0$ and the right limit is $p_{\text{\tiny MAP}}$. We will also need a space of increasing profiles
(which can be flat over some intervals)
\begin{align*}
\mathcal{S}^\prime  =\{p(\cdot)\in \mathcal{S}:\; {\rm increasing}\}
\end{align*}
and a space of strictly increasing profiles,
\begin{align*}
\mathcal{S}^{\prime\prime}  =\{p(\cdot)\in \mathcal{S}:\; {\rm strictly~increasing}\}.
\end{align*}
As will become apparent it is useful to think of profiles in $\mathcal{S}^\prime$ and $\mathcal{S}^{\prime\prime}$ as cdf's of measures over $\mathds{R}$.
Here these measures are normalized so that the measure of $\mathds{R}$ is $p_{\rm \tiny MAP}$.
Note that for $p(\cdot)\in \mathcal{S}^\prime$ the support of the associated measure is not necessarily the whole real line, while it is the whole real 
line for $p(\cdot)\in \mathcal{S}^{\prime\prime}$. Finally we introduce the space $\mathcal{S}^{\prime\prime}_0$ of strictly 
increasing profiles that are pinned at the origin, more precisely $p(\cdot)\in \mathcal{S}^{\prime\prime}_0$ if 
and only if $p(\cdot)\in \mathcal{S}^{\prime\prime}$ and $p(0) = p_{\rm \tiny MAP}/2$.

\subsection{Main Results}\label{introSettingResultsSection}
It is easy to see that $\mathcal{W}[p(\cdot)]$ is bounded
from below, more precisely $\inf_{p(\cdot)\in
\mathcal{S}}\mathcal{W}[p(\cdot)]\geq -\frac{1}{2}p_{\text{\tiny MAP}}^l$. Indeed $W_s(p)\geq 0$ as
seen in Figure~\ref{singlePotPlotat4881} and using Jensen's inequality one can
show that $\mathcal{W}_{\rm int}[p(\cdot)]\geq -\frac{1}{2}p_{\text{\tiny MAP}}^l$ (see Lemma \ref{poslem}).
The first non-trivial question one may ask is whether the minimum is
attained in $\mathcal{S}$. Using a rearrangement inequality of Brascamp-Lieb-Luttinger \cite{Brascamp} this question can be reduced 
to the same one in $\mathcal{S}^\prime\subset\mathcal{S}$. Existence in $\mathcal{S}^\prime$ will be answered by means of
the so-called ``direct method'', a standard strategy of calculus of variations \cite{Dacorogna}, \cite{Fonseca},
(see Section \ref{dirmeth}). Once we know that a minimizer exists in $\mathcal{S}^\prime$ we prove that it cannot be in 
$\mathcal{S}^\prime\setminus \mathcal{S}^{\prime\prime}$. The following existence theorem is 
proven in Section \ref{dirmeth}.

\begin{theorem}\label{existencetheorem}
Let $\epsilon=\epsilon_{\text{\tiny MAP}}$. The functional $\mathcal{W}[p(\cdot)]$ achieves its 
minimum over $\mathcal{S}$ in the subspace $\mathcal{S}^{\prime\prime}$. 
There does not exist a minimum in 
$\mathcal{S}^\prime \setminus \mathcal{S}^{\prime\prime}$.
\end{theorem}

The existence of a minimum in $\mathcal{S}^\prime\setminus\mathcal{S}^{\prime\prime}$ is excluded, i.e., a minimizer 
that increases has to be strictly increasing.
However we are not able to exclude the existence of a minimizer in $\mathcal{S} \setminus \mathcal{S}^\prime$, in other
words a minimizer that would have
``oscillations''. In order to exclude such minimizers we would have 
to study under what conditions, in our context, the Brascamp-Lieb-Luttinger inequality is strict; but we do 
not address this issue in the present work. In general this 
can be a difficult problem, see \cite{burchard}, \cite{taryn-flock}.

When functionals are convex one obtains important information on the set of minimizers. For example strict convexity implies 
that the minimizer is unique. 
Thus the next natural question is whether or not the functional $\mathcal{W}[p(\cdot)]$ is (strictly) convex. 
This is in fact not true, but we will show that it is {\em displacement convex} 
 in $\mathcal{S}^\prime$ (hence also in $\mathcal{S}^{\prime\prime}$). Here displacement convexity refers
to convexity under an interpolation path that is different from the usual linear combination.
The displacement convexity in $\mathcal{S}^\prime$ (and $\mathcal{S}^{\prime\prime}$) cannot be strict since the system certainly 
has at least one translational degree of freedom: indeed the functional is invariant under a global translation, i.e., we have
$\mathcal{W}[p(\cdot +\tau)] = \mathcal{W}[p(\cdot)]$. To investigate if displacement 
convexity is strict it is convenient to remove this degree of freedom by pinning the profiles, say at the origin. 
In fact, we will prove strict displacement convexity of the functional only in the space $\mathcal{S}^{\prime\prime}_0$ of pinned and 
strictly increasing profiles. 

\begin{theorem}\label{displacementconvtheorem}
Let $\epsilon=\epsilon_{\text{\tiny MAP}}$. The functional $\mathcal{W}[p(\cdot)]$ is displacement 
convex on $\mathcal{S}^{\prime}$, and strictly displacement convex 
on $\mathcal{S}^{\prime\prime}_0$. 
\end{theorem}

This implies that there is a unique minimizer in $\mathcal{S}^{\prime\prime}_0$. But since the 
existence of a minimizer is excluded in $\mathcal{S}^\prime \setminus \mathcal{S}^{\prime\prime}$ (by 
theorem \ref{existencetheorem}), we can conclude that 
the only minimizers in $\mathcal{S}^\prime$ are translates of the unique one in $\mathcal{S}^{\prime\prime}_0$. These 
consequences also translate into properties of solutions of the DE equation \eqref{dep}. 

\begin{corollary}\label{basiccoro}
Let $\epsilon=\epsilon_{\text{\tiny MAP}}$. In the space $\mathcal{S}^{\prime\prime}_0$ the functional $\mathcal{W}[p(\cdot)]$ 
has a unique minimizer. In $\mathcal{S}^\prime$ all minimizers are translates of it. Similarly, in  
$\mathcal{S}^{\prime\prime}_0$ the DE equation \eqref{dep}  
has a unique solution, and in $\mathcal{S}^\prime$ all solutions are translates of it. 
\end{corollary}

The proofs of Theorem \ref{displacementconvtheorem} and Corollary \ref{basiccoro} are given in Section \ref{last}.

We would like to point out that while displacement convexity itself is quite general and can presumably be generalized to 
the general potential functionals of \cite{Pfister}, the issue of strict displacement convexity is 
more subtle. In fact T. Richardson \cite{TomRich} pointed out 
 examples of system where ``internal'' translation degrees of freedom  may exist (besides the global one) which 
 would spoil the unicity up to global translations.

%


\section{Displacement Convexity} \label{displConvDefSection}
Displacement convexity can be very useful in functional
analysis. It goes back to McCann [7] and plays an important role in the theory of optimal transport \cite{Villani}. It has been used
in \cite{Alberti} and \cite{Carlen} to study a functional governing a spatially coupled Curie-Weiss model, which bears close similarities
with the coding theory model studied here (see \cite{Hassani}). In this section, we give a quick introduction to the tool of
displacement convexity.

Recall first that the usual notion of convexity of a generic
functional $\mathcal{F}[p(\cdot)]$ on a generic  space $\mathcal{X}$
means that for all $p_0(\cdot),p_1(\cdot)\in\mathcal{X}$ and $\lambda\in[0,1]$,
\begin{equation*}
\mathcal{F}[(1-\lambda) p_0(\cdot) +\lambda p_1(\cdot)]\leq\lambda\mathcal{F}[p_0(\cdot)]+(1-\lambda)\mathcal{F}[p_1(\cdot)].
\end{equation*}

Lemma \ref{rearrangementlemma} in Section \ref{existenceSection}
shows that we can restrict the minimization problem to the space
of increasing profiles. Thus, the discussion below assumes that we
consider only such profiles. This is the correct setting for defining
displacement convexity.

An increasing profile with left limit $0$ and right limit $p_{\text{\tiny
MAP}}$ can be thought of as a cdf (up to scaling because the right limit is not $1$). Further, such
increasing functions have increasing inverse functions (which can
also be thought of as cdfs, up to scaling).  More precisely, consider
the following bijective maps that associate (with an abuse of
notation) to a cdf $p(\cdot)$  its inverse $z(\cdot)$:
\begin{equation*}
\begin{split}
z(p)&=\inf\{\,z:\;p(z)>p\},\\
p(z)&=\inf\{\,p:\;z(p)>z\}. \label{displConvMaps1}
\end{split}
\end{equation*}
For any two increasing profiles $p_0(\cdot),p_1(\cdot)\in\mathcal{S}^\prime$,
we consider $z_0(\cdot),z_1(\cdot)$ their respective inverses
under the maps defined above. Then for any $\lambda \in [0,1]$, the
interpolated profile $p_{\lambda}(\cdot)$ is defined as follows:
\begin{equation*}
\begin{split}
z_{\lambda}(p)&=(1-\lambda) z_0(p)+\lambda z_1(p),\\
p_{\lambda}(z)&=\inf\{\,p:\;z_{\lambda}>z\}. \label{mapInvProf1}
\end{split}
\end{equation*}
In words, the difference in interpolation under the alternative
structure is that the linear interpolation is applied on the inverse
of the profiles of interest, and the effect of such an interpolation
is then mapped back into the space of profiles. 
It is not difficult to see that if $p_0(\cdot)$ and $p_1(\cdot)$ are in $\mathcal{S}^\prime$, 
$\mathcal{S}^{\prime\prime}$ or $\mathcal{S}^{\prime\prime}_0$, then so are the interpolating 
profiles $p_\lambda(\cdot)$ for all $\lambda\in [0,1]$.

Displacement convexity
of $\mathcal{W}[p(\cdot)]$ on the space $\mathcal{S}^\prime$ simply means
that the following inequality holds:
\begin{equation}\label{displacementconvinequ}
\mathcal{W}[p_{\lambda}(\cdot)]\leq (1-\lambda) \mathcal{W}[p(\cdot)]+\lambda\mathcal{W}[p^\prime(\cdot)]
\end{equation}
for any $p(\cdot),p^\prime(\cdot)\in\mathcal{S}^\prime$ and $\lambda\in[0,1]$. Strict displacement
convexity means that this inequality is strict as long as $p_0(\cdot)$
and $p_1(\cdot)$ are distinct and $\lambda\in ]0,1[$.  We will prove 
displacement convexity of $\mathcal{W}[p(\cdot)]$ by separately
proving this property, in Sections \ref{displConvSingleSection} and
\ref{displConvIntSection}, for the two functionals
(\ref{singleTermIntegral1}) and (\ref{interactionTerm1}), respectively. Moreover we will see that \eqref{interactionTerm1}
is strictly displacement convex in $\mathcal{S}^{\prime\prime}_0$.

\section{Existence of Minimizing Profile} \label{existenceSection}
In this section, we prove that the functional $\mathcal{W}$ attains
its minimum.

\subsection{Preliminaries}\label{prel}

We start by some preliminaries to show that one can
restrict the search of minimizing profiles to those in $\mathcal{S}$ that are monotone increasing.
The proofs of the lemmas can be found in
Appendices \ref{positivitySection}-\ref{tightnessSection}.

The first
lemma states that the interaction potential is bounded from below.
\begin{lemma}\label{poslem}
For any $p(\cdot)$ in $\mathcal{S}$,
{\begin{equation*}
\int\limits_{\mathds{R}}\,\mathrm{d}z\; \Bigg\{ p(z)^{l}-\Bigg( \int\limits_{0}^{1}\,\mathrm{d}u\;p(z+u) \Bigg)^{l} \Bigg\}\geq -\frac{1}{2}p_{\text{\tiny MAP}}^{l}.
\label{positivityLemma}
\end{equation*}}
\end{lemma}
\begin{proof}
See Appendix \ref{positivitySection}.
\end{proof}

We remark that any constant lower bound is sufficient for our purposes: finding 
profiles that minimize a potential is equivalent to find those that minimize a 
potential added to a constant.
The following
lemma states that a truncation of the profile at the value $p_{\text{\tiny MAP}}$
decreases the potential functional, so we may restrict our search of minimizing
profiles to those with range $p(z) \in [0,p_{\text{\tiny MAP}}]$.
\begin{lemma}\label{truncationLemma}
Define $\bar{p}(z)= \min\{p(z), p_{\text{\tiny MAP}}\}$.
For all $p(\cdot)\in\mathcal{S}$ we have
\begin{equation*}
\mathcal{W}[p(\cdot)] \geq \mathcal{W}[\bar{p}(\cdot)],
\end{equation*}
and the inequality is strict if $p(\cdot)\neq\bar{p}(\cdot)$.
\end{lemma}
\begin{proof}
See Appendix \ref{truncationSection}.
\end{proof}

We next restrict our search of minimizing profiles to increasing
ones. In order to achieve this we will use rearrangement inequalities.
Here we will need a notion of increasing rearrangement, see \cite{Alber}. In
words, an increasing rearrangement associates to any function $p(\cdot)\in \mathcal{S}$ with range $[0, p_{\text{\tiny MAP}}]$ an increasing function
$p^*(\cdot)\in \mathcal{S}^\prime$ so that the total mass is preserved. More formally,
any non-negative function in $\mathcal{S}$ can be represented in layer cake form
\begin{equation*}
 p(z) = \int_0^{+\infty} dt\, \mathbb{1}_{E_t}(z),
\end{equation*}
where $\mathbb{1}_{E_t}(z)$ is the indicator function of the level
set $E_t=\{z\vert p(z) > t\}$.  For each $t$, the level set $E_t$
can be written as the union of a bounded set $A_t$ and a half line
$]a_t, +\infty[$. We define the rearranged set $E_t^* = ] a_t- \vert
A_t\vert, +\infty[$. The increasing rearrangement of $p(\cdot)$ is
the new function $p^*(\cdot)$ whose level sets are $E_t^*$. More
explicitly,
\begin{equation*}
p^*(z) = \int_0^{+\infty} dt\, \mathbb{1}_{E_t^*}(z).
\end{equation*}

\begin{lemma}\label{rearrangementlemma}
Take any $p(\cdot)\in \mathcal{S}$ and let $p^*(\cdot)\in \mathcal{S}^\prime$ be its increasing rearrangement. Then,
\begin{equation*}\label{incrRearrangementLemma}
\mathcal{W}[p(\cdot)] \geq \mathcal{W}[p^*(\cdot)].
\end{equation*}
\end{lemma}
\begin{proof}
See Appendix \ref{incrRearrangement}.
\end{proof}

We can thus restrict the search of minimizing profiles to the space of increasing profiles (but as already explained before we cannot exclude that 
there exist a $p(\cdot)\in \mathcal{S}$ such that $\mathcal{W}[p(\cdot)] = \mathcal{W}[p^*(\cdot)]$). In fact, the following lemma
allows to further restricts the search of minimizing profiles to strictly increasing ones. 
\begin{lemma}\label{strIncrProp}
Let $p(\cdot)\in \mathcal{S}^\prime$ be a minimizer of the potential functional $\mathcal{W}[p(\cdot)]$ that is in $\mathcal{S}^\prime$. 
Then it must be strictly increasing, i.e., $p(\cdot)\in \mathcal{S}^{\prime\prime}$.
\end{lemma}
\begin{proof}
See Appendix \ref{strIncrSection}.
\end{proof}

The final step of these preliminaries concerns a necessary condition that any {\it minimizing sequence} in $\mathcal{S}^{\prime\prime}_0$
must satisfy. It is useful to think of such profiles as cdf's. A minimizing sequence in $\mathcal{S}^{\prime\prime}_0$
is by definition any sequence $p_n(\cdot)\in \mathcal{S}^{\prime\prime}_0$ such that
\begin{equation}\label{minSequence}
 \lim_{n\rightarrow\infty}\mathcal{W}[p_{n}(\cdot)]=\inf_{p\in \mathcal{S}^{\prime\prime}_0}\mathcal{W}[p(\cdot)].
\end{equation}
Such a sequence exists as long as the functional is bounded from below. Since $W_s(p)\geq 0$ and due to Lemma \ref{poslem}, this is true.
Consider the sequence of probability measures associated to the sequence of cdfs $p_n(\cdot)$. The following lemma states
that this sequence of measures is {\it tight}.

\begin{lemma}\label{tightnesslemma}
 Let $p_n(\cdot)\in \mathcal{S}^{\prime\prime}_0$ be a minimizing sequence of cdfs. For any $\delta>0$ we can find $M_\delta > 0$ (independent of $n$) such that
 \begin{equation*}\label{spaceSdprime}
 p_n(M_\delta)-p_n(-M_\delta) > (1-\delta)p_{\text{\tiny MAP}}
\end{equation*}
for all $n$.
\end{lemma}
\begin{proof}
 See Appendix \ref{tightnessSection}.
\end{proof}

\subsection{The Direct Method}\label{dirmeth}
The direct method in the calculus of variations
\cite{Dacorogna}-\cite{Fonseca} is a standard scheme to prove that
minimizers exist. We use this method to obtain the following theorem:
\begin{proof}[Proof of theorem \ref{existencetheorem}]
Let us take any minimizing sequence $p_{n}(\cdot)$ of cdfs i.e., a
sequence that satisfies (\ref{minSequence}). By Lemma \ref{tightnesslemma}
the corresponding sequence of measures is tight. Thus by a simple
version of Prokhorov's theorem for measures on the real line, we
can extract a (point-wise) convergent subsequence of cdfs $p_{n_k}(\cdot)
\to p_{\ell}(\cdot)$ as $k\to +\infty$ with $p_{\ell}(\cdot)\in \mathcal{S}^{\prime\prime}_0$.
By Fatou's Lemma, one can check that the potential functional
is lower-semi-continuous, which means
\begin{equation}\label{LSC}
\mathcal{W}[p_{\ell}(\cdot)]\leq \liminf_{k\rightarrow +\infty}\mathcal{W}[p_{n_k}(\cdot)].
\end{equation}
Putting \eqref{minSequence} and \eqref{LSC} together,
\small
\begin{equation*}
 \mathcal{W}[p_{\ell}(\cdot)]\leq \liminf_{k\rightarrow +\infty}\mathcal{W}[p_{n_k}(\cdot)]
 = \lim_{n\rightarrow +\infty}\mathcal{W}[p_{n}(\cdot)]
 = \inf_{\mathcal{S}^{\prime\prime}_0}\mathcal{W}[p(\cdot)],
\end{equation*}
\normalsize
On the other hand $\inf\limits_{\mathcal{S}^{\prime\prime}_0}\mathcal{W}[p(\cdot)]\leq \mathcal{W}[p_{\ell}(\cdot)]$. Thus, we conclude
that $\inf\limits_{\mathcal{S}^{\prime\prime}_0}\mathcal{W}[p(\cdot)] = \mathcal{W}[p_{\ell}(\cdot)]$.

We have shown that the minimum is achieved in $\mathcal{S}^{\prime\prime}_0$ the space 
of strictly increasing profiles pinned at the origin. Hence it is achieved in $\mathcal{S}^{\prime\prime}$
and $\mathcal{S}^{\prime}$ (note that by translation invariance, translations of $p_\ell(\cdot)$ are minimizers in these spaces). 
Finally, Lemma \ref{strIncrProp} ensures that there is no minimum in $\mathcal{S}^\prime \setminus \mathcal{S}^{\prime\prime}$.
\end{proof}

\section{Analysis of Displacement Convexity for the Functional $\mathcal{W}[p(\cdot)]$} \label{displConvSection}

This section contains the main results of the paper, namely that
the potential functional $\mathcal{W}[p(\cdot)]$ is displacement convex in $\mathcal{S}^\prime$ and strictly 
displacement convex in $\mathcal{S}^{\prime\prime}_0$. 


\subsection{Displacement Convexity of the Single-Potential Term}\label{displConvSingleSection}

We first prove that the single-potential functional $\mathcal{W}_{\rm
single}[p(\cdot)]$ is displacement convex. 
Note that the single system potential $W_s(p)$ is {\it not convex} in
the usual sense (see Figure \ref{singlePotPlotat4881}).
\begin{proposition}\label{displacementconvsingle}
Let $p_0(\cdot)$ and $p_1(\cdot)$ be in $\mathcal{S}^\prime$ and let
$p_{\lambda}(\cdot)$ the interpolating profile as defined in Section
\ref{mapInvProf1}. Then
\begin{equation*}
\mathcal{W}_{\rm single}[p_{\lambda}(\cdot)] = (1-\lambda)\mathcal{W}_{\rm single}[p_0(\cdot)]+\lambda\mathcal{W}_{\rm single}[p_1(\cdot)].
\end{equation*}
\end{proposition}
\begin{proof}
Recall that $\mathcal{W}_{\rm
single}[p(\cdot)]=\int_{\mathds{R}}\mathrm{d}z W_{s}(p(z))$. Recall
also that $p_{\lambda}(z)$ as defined in Section \ref{mapInvProf1} is the inverse
of $z_{\lambda}(p) = (1-\lambda) z_0(p) + \lambda z_1(p)$.  Thus
\begin{equation*}
\begin{split}
\int_{\mathds{R}}&\mathrm{d}z W_{s}(p_{\lambda}(z)) = \int_{0}^{p_{\text{\tiny MAP}}}\mathrm{d}z_{\lambda}(p)W_{s}(p)\\
&=(1-\lambda)\int_{0}^{p_{\text{\tiny MAP}}}\mathrm{d}z_0(p)W_{s}(p) + \lambda\int_{0}^{p_{\text{\tiny MAP}}}\mathrm{d}z_1(p)W_{s}(p)\\
&=(1-\lambda)\int_{\mathds{R}}\mathrm{d}zW_{s}(p_0(z))+\lambda\int_{\mathds{R}}\mathrm{d}zW_{s}(p_1(z)).
\end{split}
\end{equation*}
\end{proof}
Thus, the function $\lambda\rightarrow\mathcal{W}_{\rm
single}[p_\lambda(\cdot)]$ is linear, hence convex.

\subsection{Displacement Convexity of the Interaction-Potential Term}\label{displConvIntSection}
The proof of displacement convexity of the interaction
potential term is more involved.
\begin{proposition}\label{displConvIntTheorem} Let $p_0(\cdot)$ and
$p_1(\cdot)$ be in $\mathcal{S}^\prime$ and let $p_{\lambda}(\cdot)$ the
interpolating profile as defined in Section \ref{mapInvProf1}. Then 
\begin{equation}\label{dc}
 \mathcal{W}_{\rm int}[p_{\lambda}(\cdot)] \leq (1-\lambda)\mathcal{W}_{\rm int}[p_0(\cdot)]+\lambda\mathcal{W}_{\rm int}[p_1(\cdot)].
\end{equation}
\end{proposition}
\begin{proof}
Since $p$ can be seen as a cdf we associate with it its probability
measure $\mu$ such that $p(z)=p_{\rm \small
MAP}\int_{-\infty}^{z}\mathrm{d}\mu(x)$.  Let us rewrite
the  interaction functional in the form
\begin{equation}\label{VillanisForm}
\displaystyle \mathcal{W}_{\rm int}[p_\lambda(\cdot)]=\int_{\mathds{R}}\mathrm{d}\mu_\lambda(x_{1})\dots
\mathrm{d}\mu_\lambda(x_{l})V(x_{1},\dots,x_{l}),
\end{equation}
where $V(x_1,\dots, x_l)$ is a totally symmetric ``kernel function''
that we will compute.  There is an argument (see \cite{Villani})
that allows to conclude \eqref{dc} whenever $V$ is jointly convex
(in the usual sense). Let us briefly explain this argument here.
Consider the measures $\mu_0$, $\mu_1$ associated to cdfs
$p_0(\cdot)$, $p_1(\cdot)$. Then there exists a unique increasing map
$T:\mathds{R}\to \mathds{R}$ such that $\mu_1 = T\#\mu_0$. Here
$T\#\mu_0$ is the push-forward\footnote{Given a measurable map
$T:\mathds{R}\rightarrow\mathds{R}$, the push-forward of $\mu$ under
$T$ is the measure $T\#\mu_0$ such that, for any bounded continuous
function $\phi$,
$\int\limits_{\mathds{R}}\phi(T(x))\mathrm{d}\mu(x)=\int\limits_{\mathds{R}}\phi(x)\mathrm{d}(T\#\mu)(x)$.}
of $\mu_0$ under $T$.  Then from $x_\lambda(p) = (1-\lambda) x_0(p) +
\lambda x_1(p)$ we have that $\mu_\lambda = T_\lambda\#\mu_0$
where $T_\lambda(x) = \lambda x + (1-\lambda) T(x)$. Equation
\eqref{VillanisForm} can be written as
\begin{align}
& \mathcal{W}_{\rm int}[p_\lambda(\cdot)] =\int_{\mathds{R}}\mathrm{d}\mu_0(x_{1})\dots
\mathrm{d}\mu_0(x_{l})V(T_\lambda(x_{1}),\dots,T_\lambda(x_{l}))
\nonumber \\ &
= l!\int_{S_x}\mathrm{d}\mu(x_{1})\dots
\mathrm{d}\mu(x_{l})V(T_\lambda(x_{1}),\dots,T_\lambda(x_{l})).
\label{againstmeasures}
\end{align}
In the second equality we restrict the integrals over the sector
$S_{x}=\{\mathbf{x}=(x_1,\cdots,x_l):\,x_i\geq x_j \;\mathrm{if}\;i<j \}$, which is possible since $V$ is
totally symmetric. Now it is important to notice that since $T$ is
an increasing map we have $T_\lambda(x_{1})\geq\dots \geq T_\lambda(x_{l})$ for any $\lambda\in [0,1]$. Moreover
the $\lambda$ dependence in the kernel function is linear. Thus the
proof of displacement convexity ultimately rests on checking that
the kernel function is jointly convex in one sector, say $S_x$.\footnote{By symmetry, convexity in one
sector implies convexity in other sectors. However this does not
mean that convexity holds if arguments are taken in different
sectors. And, indeed in the present problem one can check that
convexity only holds within each sector.} In fact the kernel function is translation invariant and can be expressed
as a function of the distances $d_{1i}\equiv x_1 - x_i$, $i=1,\dots, l$. We will prove joint convexity of $V$ as a function of these distances. 

Now it remains to compute $V$ and to investigate its joint convexity.  With
appropriate usage of Fubini's theorem and after some manipulations,
we find
\begin{align}
&\mathcal{W}_{\rm int}[p(\cdot)]=
\frac{\epsilon p_{\text{\tiny MAP}}^{l}}{l}\int_{\mathds{R}^l}\prod_{i=1}^{l}\mathrm{d}\mu_0(x_{i})\Biggl\{\int_{[0,1]^{l}}
\prod_{i=1}^{l}\mathrm{d}u_{i}\notag \\
& \quad\int_{\mathds{R}}\mathrm{d}z
\bigg(\prod_{i=1}^l\theta(z-x_{i})-\prod_{i=1}^l\theta(z-(x_{i}-u_{i}))\bigg)\biggr\}, \label{intTermMu}
\end{align}
where $\theta(x)$ denotes the Heaviside step function. So the kernel $V(x_1,\dots,x_l)$ in (\ref{VillanisForm}) is the
integrand of the first $l$ integrals in (\ref{intTermMu}). Our goal
henceforth is to prove that $V$ is convex in the usual sense.
We will prove that in fact $V_{\mathbf{u}}$
is convex for all fixed $\mathbf{u}$, where $\mathbf{u} =(u_1,\dots , u_l)$, 
\begin{equation*}
\begin{split}
&V_{\mathbf{u}}(\mathbf{x})=\int_{\mathds{R}}dz
\biggl(\prod_{i=1}^l\theta(z-x_{i})-\prod_{i=1}^l\theta(z-(x_{i}-u_{i}))\biggr).\label{subproblem}
\end{split}
\end{equation*}
We recall here that we restrict our analysis to the sector of 
the space of variables $S_x$. Also, we remark that
$\prod_{i=1}^{l}\theta(a_{i})=\theta(\max\limits_{i=1\dots l}a_{i})$.
We observe that
$V_{\mathbf{u}}$ can be written in terms of the distances $d_{1i}= x_1 - x_i$, $i=2,\dots, l$ as (here $d_{1i}\equiv 0$)
\begin{equation*}
\begin{split}
V_{\mathbf{u}}(\mathbf{x})&=\int_{\mathds{R}}\mathrm{d}z\{\theta(z-x_{1})-\theta(\max\limits_{i=1\dots l}(z-(x_{i}-u_{i}))) \}\\
&=-\min\limits_{i=1\dots l}(x_{1}-x_{i}+u_{i}) = -\min\limits_{i=1\dots l}(d_{1i}+u_{i}) \label{VuInD}
\end{split}
\end{equation*}
Lemma \ref{conclem} below states that $V_{\mathbf{u}}$ is jointly convex in $S_x$ for all fixed $\mathbf{u}$. This implies that $V(\mathbf{x})$ is
jointly convex in $S_x$. 
This completes the proof.
\end{proof}
\begin{lemma}\label{conclem}
The function $f_{\mathbf{u}}(\mathbf{d})=\min\limits_{i}(d_{1i}+u_{i})$ is concave
in $\mathbf{d}$, where $\mathbf{d}=(d_{12},\dots,d_{1l})$ and $d_{1i}\equiv 0$. 
\end{lemma}
\begin{proof}
Let $\mathbf{d}$ and $\mathbf{d}'$ be two instances of the argument of $f_{\mathbf{u}}$. Then, for $\lambda\in[0,1]$,
\begin{equation*}
\begin{split}
f_{\mathbf{u}}((1-\lambda)\mathbf{d}+&\lambda\mathbf{d}')=\min\limits_{i}((1-\lambda) d_{1i}+\lambda d_{1i}'+u_{i})\\
&=\min\limits_{i}((1-\lambda) (d_{1i}+u_{i})+\lambda(d_{1i}'+u_{i}))\\
&\geq (1-\lambda)\min\limits_{i}(d_{1i}+u_{i})+\lambda\min\limits_{i}(d_{1i}'+u_{i})\\
&= (1-\lambda) f_{\mathbf{u}}(\mathbf{d})+\lambda f_{\mathbf{u}}(\mathbf{d}').
\end{split}
\end{equation*}
This shows concavity. 

\end{proof}

\subsection{Strict displacement convexity}\label{substrict}

We now prove that for $p_0(\cdot)$ and $p_1(\cdot)$ in $\mathcal{S}^{\prime\prime}_0$ the inequality \eqref{dc} is strict whenever $\lambda \neq 0,1$. 

Since we already know that the kernel function $V(\mathbf{x})$ is convex 
in the sector $S_x$, it is sufficient to show it is {\it strictly convex as a function of the distances} 
$\mathbf{d}=(d_{12}, \dots, d_{1l})$
on some {\it subset of positive measure} of the $\mathbf{d}$-space. This is sufficient because in 
\eqref{againstmeasures} the kernel function is integrated against measure $\mu_0$ with full support $\mathds{R}$. Recall that 
$\mu_0$ has full support because the profile $p_0(\cdot)$ belongs to $\mathcal{S}^{\prime\prime}$. 

In Appendix \ref{explicitkernel} we give explicit formulas for $V$ in terms of 
the distances $d_{1i}$. These formulas allow to prove that $V$
is strictly convex in a subset of non-zero measure. Concretely, this subset is a small enough neighborhood of the origin $d_{1i}=0$, $i=2,\dots,l$. 

We remark that it is not easy to see that $V$ is convex in the whole $\mathbf{d}$-space directly from these formulas (however we already know that $V$ is convex by the method of proof of the previous section). In fact the formulas show that it is certainly not strictly convex when some of the distances become greater than $1$.

\section{Proofs of theorem \ref{displacementconvtheorem} and corollary \ref{basiccoro}}\label{last}

In this paragraph, for completeness, we wrap up the proofs of Theorem \ref{displacementconvtheorem} and Corollary \ref{basiccoro}.

\begin{proof}[Proof of Theorem \ref{displacementconvtheorem}]
Displacement convexity of $\mathcal{W}[p(\cdot)]$ in $\mathcal{S}^\prime$ follows from Propositions \ref{displacementconvsingle}
and \ref{displConvIntTheorem}. In section \ref{substrict} we have shown that $\mathcal{W}_{\rm\tiny int}[p(\cdot)]$ is strictly displacement convex in $\mathcal{S}^{\prime\prime}_0$. Combining this with  Proposition \ref{displacementconvsingle} immediately yields 
strict displacement convexity of $\mathcal{W}[p(\cdot)]$ in $\mathcal{S}^{\prime\prime}_0$.
\end{proof}

\begin{proof}[Proof of Corollary \ref{basiccoro}]
Uniqueness of the minimizer of $\mathcal{W}[p(\cdot)]$ in $\mathcal{S}^{\prime\prime}_0$ follows from strict displacement convexity. 
Indeed suppose there are two distinct minimizers $p_1(\cdot)$ and $p_2(\cdot)$ with 
 $\mathcal{W}[p_1(\cdot)] = \mathcal{W}[p_2(\cdot)]$
and let $p_\lambda(\cdot)$ be the displacement interpolant. Then
$\mathcal{W}[p_\lambda(\cdot)] < (1-\lambda)\mathcal{W}[p_1(\cdot)] + \lambda \mathcal{W}[p_2(\cdot)]$ (for $\lambda \neq 0, 1$) which implies 
$\mathcal{W}[p_\lambda(\cdot)] < \mathcal{W}[p_1(\cdot)]$. We must also have 
$\mathcal{W}[p_1(\cdot)]\leq \mathcal{W}[p_\lambda(\cdot)]$, hence $\mathcal{W}[p_1(\cdot)] < \mathcal{W}[p_\lambda(\cdot)]$
which is a contradiction. 

Let us now show that all minimizers $p_1(\cdot)\in \mathcal{S}^\prime$ are translates of the unique 
minimizer $p_{0}(\cdot)\in \mathcal{S}^{\prime\prime}_0$. We know from Theorem \ref{strIncrProp} that 
$p_1(\cdot)\in \mathcal{S}^{\prime\prime}$, i.e., it has to be strictly 
increasing. Thus there is a unique position, say $z_1$ such that $p_1(z_1) = p_{\rm\tiny MAP}/2$. Consider the set of profiles 
$\mathcal{S}^{\prime\prime}_{z_1}$ obtained by translating the set $\mathcal{S}^{\prime\prime}_{0}$ by the vector $z_1$. Clearly 
$p_1(\cdot)$ is the unique minimizer in $\mathcal{S}^{\prime\prime}_{z_1}$. But it is also clear that $\mathcal{W}[p_{0}(\cdot - z_1)] = 
\mathcal{W}[p_1(\cdot)]$. Thus, since $p_{0}(\cdot - z_1)\in \mathcal{S}^{\prime\prime}_{z_1}$, we must have 
$p_{0}(\cdot - z_1)=p_1(\cdot)$ as announced.

Finally let us discuss the consequences for the solutions of the DE equation \eqref{dep}. We show that in the space $\mathcal{S}^\prime$ 
a solution of the DE equation is necessarily a minimum of $\mathcal{W}[p(\cdot)]$. This implies the statement of the theorem.
Let $p_0(\cdot)\in \mathcal{S}^\prime$ a solution of the DE equation. Consider any other profile $p_1(\cdot)\in \mathcal{S}^\prime$, 
and consider the 
displacement interpolant $p_\lambda(\cdot)$. A computation of the derivative shows 
that $\frac{d}{d\lambda}\mathcal{W}[p_\lambda(\cdot)]\vert_{\lambda=0} = 0$ because $p_0(\cdot)$ is a solution of the DE equation. 
Since the map $\lambda \to \mathcal{W}[p_\lambda(\cdot)]$
is convex, $\lambda=0$ must be a minimum of this map. Thus $\mathcal{W}[p_\lambda(\cdot)]\geq \mathcal{W}[p_0(\cdot)]$
and in particular with $\lambda =1$ we get $\mathcal{W}[p_1(\cdot)]\geq \mathcal{W}[p_0(\cdot)]$. Thus $p_0(\cdot)$ is a minimum 
of the functional in $\mathcal{S}^\prime$.
\end{proof}

\section{Conclusion}
In this paper, we demonstrate a new tool for the analysis of spatially
coupled codes, namely the concept of displacement convexity. This
tool makes use of an alternative structure of probability distributions
and hence applies to an appropriate space of increasing profiles.
We prove that the potential functional governing the $(l,r)$-regular
ensemble is (strictly) convex
under the alternative structure. This result implies that the
potential functional admits a unique minimizing profile, or
equivalently, that the DE equations governing the system admit a
unique FP solution, in an appropriate space of profiles.

There are several questions that can be posed in this context.
First, we recall that the original potential functional governing
the system at hand is in discrete form. Can one extend the displacement
convexity framework to the discrete setting?  Displacement convexity
can presumably be used to analyze a large range of problems with
flavors similar to the present one. The generalization to irregular LDPC ensembles 
is immediate for pure displacement convexity (the question of strict convexity is however more subtle). It is 
interesting to consider more general one dimensional
scalar recursions as in \cite{Pfister} and find out what are the general restrictions
on the single system potential that still allow to prove displacement (strict) convexity.
It also remains to be seen if these techniques can be applied 
 to general BMS channels, the random K-SAT and Q-coloring
problems to name a few. We plan to come back to these problems in
the future.

\appendices

\section{Lower Boundedness of the Interaction Potential} \label{positivitySection}

\begin{proof}[Proof of Lemma \ref{poslem}]
From Jensen's inequality, 
\begin{equation}\label{jensen}
 \displaystyle \int\limits_{0}^{1}\mathrm{d}u\;p(z+u)^{l}\geq\Big( \int\limits_{0}^{1}\mathrm{d}u\; p(z+u)\Big)^{l}.
\end{equation}
Further,
\begin{align*}
\displaystyle &\int\limits_{-M}^{M}\mathrm{d}z\int\limits_{0}^{1}\mathrm{d}u\;p(z+u)^{l}\overset{(a)}{=}\int\limits_{0}^{1}\mathrm{d}u\int\limits_{-M+u}^{M+u}\mathrm{d}z'\;p(z')^{l}\\
&=\int\limits_{0}^{1}\mathrm{d}u\Big(\int\limits_{-M+u}^{-M}\mathrm{d}z'p(z')^{l}+\int\limits_{-M}^{M}\mathrm{d}z'p(z')^{l}+\int\limits_{M}^{M+u}\mathrm{d}z'p(z')^{l}\Big),
\end{align*}
where (a) is obtained by first changing the order of integration
(which is admissible since the integral converges) and
then making the change of variable $z'=z+u$. And so, by combining
this identity with \eqref{jensen} we obtain
\begin{equation*}
\begin{split}
\displaystyle &\int\limits_{-M}^{M}\mathrm{d}z\;\Bigg\{p(z)^{l}-\Big( \int\limits_{0}^{1}\mathrm{d}u\; p(z+u)\Big)^{l}\Bigg\}\\
&\qquad+\int\limits_{0}^{1}\mathrm{d}u\int\limits_{-M+u}^{-M}\mathrm{d}z'\;p(z')^{l}+
\int\limits_{0}^{1}\mathrm{d}u\int\limits_{M}^{M+u}\mathrm{d}z'\;p(z')^{l}\geq 0.
\end{split}
\end{equation*}
Now we take the limit $M\to +\infty$ for each term of this inequality.
By an application of Lebesgue's dominated convergence theorem, the
last two terms tend to zero and $\frac{1}{2}p_{\text{\tiny MAP}}^{l}$, respectively. 
Therefore the limit of the first term is bounded from below by $-\frac{1}{2}p_{\text{\tiny MAP}}^{l}$,
which concludes the proof.  \end{proof}

\section{Truncation of Profiles} \label{truncationSection}

\begin{proof}[Proof of Lemma \ref{truncationLemma}]
It is easy to prove that a truncation of $p(z)$ at $p_{\text{\tiny
MAP}}$ yields a smaller value for the single system potential
$\mathcal{W}_s(p(z))$ (see e.g. the Figure \ref{singlePotPlotat4881}
for an intuition).  Therefore we have $\mathcal{W}_{\rm
single}[p(\cdot)]\geq \mathcal{W}_{\rm single}[\bar p(\cdot)]$.

We now treat the functional corresponding to the interaction term. We define the function
$g$ as $g(z)=p(z)-\bar{p}(z)$ and notice that:
\begin{equation}
\begin{cases}
p(z)\leq p_{\text{\tiny MAP}} \Rightarrow g(z)=0 \; \mathrm{and}\; \bar{p}(z)=p(z),\\
p(z)> p_{\text{\tiny MAP}} \Rightarrow g(z)>0 \; \mathrm{and}\; \bar{p}(z)=p_{\text{\tiny MAP}}. \label{propertiesG}
\end{cases}
\end{equation}

We need to show that $\mathcal{W}_{\rm int}[\bar{p}(\cdot)]\leq\mathcal{W}_{\rm int}[p(\cdot)]$, or equivalently that:
\small
\begin{equation*}
\begin{split}
\displaystyle &\int\limits_{\mathds{R}}\mathrm{d}z\,\Big\{ \bar{p}(z)^{l}-\Big( \int\limits_{0}^{1}\mathrm{d}u\, \bar{p}(z+u) \Big)^{l} \Big\}\\
&\;\;\leq\int\limits_{\mathds{R}}\mathrm{d}z\,\Big\{ (\bar{p}(z)+g(z))^{l}-\Big( \int\limits_{0}^{1}\mathrm{d}u\,(\bar{p}(z+u)+g(z+u)) \Big)^{l} \Big\}.
\end{split}
\end{equation*}
\normalsize
Using the binomial expansion this is equivalent to
\begin{align*}
\sum\limits_{i=0}^{l-1} & {l \choose i} \int\limits_{\mathds{R}}\mathrm{d}z\,\Bigg\{ \bar{p}(z)^{i}g(z)^{l-i}
\\ &
-\Big( \int\limits_{0}^{1}\mathrm{d}u\,\bar{p}(z+u) \Big)^{i} \Big( \int\limits_{0}^{1}\mathrm{d}u\,g(z+u) \Big)^{l-i} \Bigg\} 
\geq 0.
\end{align*}
In the following steps, we show that the integral inside the summation above is positive for any fixed value of $i$; the inequality follows directly. We see that:
\begin{equation*}
\begin{split}
\displaystyle \Big(& \int\limits_{0}^{1}\mathrm{d}u\,\bar{p}(z+u) \Big)^{i} \Big( \int\limits_{0}^{1}\mathrm{d}u\,g(z+u) \Big)^{l-i}\\
&\leq p_{\text{\tiny MAP}}^{i}\Big( \int\limits_{0}^{1}\mathrm{d}u\,g(z+u) \Big)^{l-i}\leq p_{\text{\tiny MAP}}^{i}\int\limits_{0}^{1}\,\mathrm{d}u\,g(z+u)^{l-i} 
\label{lessThan1},
\end{split}
\end{equation*}
where the first inequality is due to the property $\bar{p}(z)\leq
p_{\text{\tiny MAP}}$ and the second is using the convexity of the
function $f(g)=g^{l};\;g\geq 0$. We integrate over $z$, then make
the change of variable $z'=z+u$ on the right-hand side to obtain:
\begin{equation}
\begin{split}
\displaystyle &\int\limits_{\mathds{R}}\,\mathrm{d}z\, \Big( \int\limits_{0}^{1}\mathrm{d}u\,\bar{p}(z+u) \Big)^{i} \Big( \int\limits_{0}^{1}\mathrm{d}u\,g(z+u) \Big)^{l-i}\\
&\leq \int\limits_{\mathds{R}}\,\mathrm{d}z\, p_{\text{\tiny MAP}}^{i}\int\limits_{0}^{1}\,\mathrm{d}u\,g(z+u)^{l-i} = 
\int\limits_{\mathds{R}}\,\mathrm{d}z'\,p_{\text{\tiny MAP}}^{i}g(z')^{l-i}
\end{split}\label{differenceG}
\end{equation}
Using the properties of $g$ in (\ref{propertiesG}), we remark that
$\int\limits_{\mathds{R}}\,\mathrm{d}z\,p_{\text{\tiny MAP}}^{i}g(z)^{l-i}
= \int\limits_{\mathds{R}}\,\mathrm{d}z\,\bar{p}(z)^{i}\,g(z)^{l-i}
$ and so the difference of quantities in the inequality (\ref{differenceG})
is integrable, and thus we obtain:
\small
\begin{equation*}
\displaystyle \int\limits_{\mathds{R}}\mathrm{d}z \Big( \int\limits_{0}^{1}\mathrm{d}u\,\bar{p}(z+u) \Big)^{i} \Big( \int\limits_{0}^{1}\mathrm{d}u\,g(z+u) \Big)^{l-i}\leq \int\limits_{\mathds{R}}\mathrm{d}z\,\bar{p}(z)^{i}g(z)^{l-i}
\end{equation*}
\normalsize
for any $i$. This yields the desired result $\mathcal{W}_{\rm int}[p(\cdot)]\geq \mathcal{W}_{\rm int}[\bar p(\cdot)]$.
\end{proof}

\section{Rearrangement of Profiles} \label{incrRearrangement}

Before proceeding with the proof of Lemma \ref{rearrangementlemma}, we state a general rearrangement inequality 
of Brascamp, Lieb and Luttinger \cite{Brascamp}.
\begin{theorem}\label{BrascampsTheorem}
Let $f_{j},\;1<j<k$ be nonnegative measurable functions on $\mathds{R}$, and let $a_{jm},\;1<j<k,\;1<m<n$, be real numbers. 
Then, if ${f}^*$ is the symmetric decreasing rearrangement of $f$, we have:
\begin{equation}
\int\limits_{\mathds{R}^{n}}\mathrm{d}^{n}x\prod\limits_{j=1}^{k}f_{j}
\Bigg( \sum\limits_{m=1}^{n}a_{jm}x_{m}\Bigg)\leq \int\limits_{\mathds{R}^{n}}\mathrm{d}^{n}x\prod\limits_{j=1}^{k}{f}^{*}_{j}
\Bigg( \sum\limits_{m=1}^{n}a_{jm}x_{m}\Bigg)
\end{equation} 
\end{theorem}

\begin{remark}
Theorem \ref{BrascampsTheorem} is nontrivial only if $k>n$. Otherwise, both integrals 
diverge and the inequality trivially holds. We will see in this section that $k>n$ in our case.
\end{remark}

\begin{proof}[Proof of Lemma \ref{incrRearrangementLemma}]
It is sufficient to prove that the increasing rearrangement of a 
profile decreases $\mathcal{W}_{\rm int}[p(\cdot)]$, since $\mathcal{W}_{\rm single}[p(\cdot)]$ is invariant 
under rearrangement. 

Theorem \ref{BrascampsTheorem} applies to symmetric decreasing rearrangements. Therefore it is convenient to first
``symmetrize'' the profile and the functional.
Consider a profile $p(\cdot)\in\mathcal{S}$
such that $p(z)\in[0,p_{\text{\tiny MAP}}]$ (due to Lemma \ref{truncationLemma}) and denote by $\hat{p}(\cdot)$
the function such that $\hat{p}(z)=p(z),\;z<R$ and $\hat{p}(z)=\hat{p}(2R-z),\;z>R$. The value $R$ is chosen (large enough) so that 
$p(R)$ is arbitrarily close to $p_{\text{\tiny MAP}}$. Note that $\hat p(\cdot)$ is integrable over $\mathds{R}$.


We recall the expression of $\mathcal{W}_{\rm int}[p(\cdot)]$ in (\ref{interactionTerm1}) and rewrite it as:
\begin{align}
\mathcal{W}_{\rm int}&[p(\cdot)] 
 \nonumber
 \\ 
= \frac{\epsilon}{l}&\lim_{R\rightarrow +\infty}\biggr\{\int\limits_{-\infty}^{R}\,\mathrm{d}z p(z)^{l}-\int\limits_{-\infty}^{R}\,\mathrm{d}z
\bigl( \int\limits_{0}^{1}\,\mathrm{d}u\;p(z+u) \bigr)^{l}\biggr\}.
\label{limFRP}
\end{align}
We now express both integrals in the bracket in terms of the symmetrized profile. For the first one, this is immediate 
\begin{align}\label{rear1}
 \int_{-\infty}^{R}\mathrm{d}z p(z)^{l} = \frac{1}{2}\int_{\mathds{R}}\mathrm{d}z \hat p(z)^{l}.
\end{align}
For the second one, some care has to be taken with the averaging over $u$ when $z$ is near $R$. One has
\begin{align}
 \int_{-\infty}^{R}\mathrm{d}z & \bigl( \int\limits_{0}^{1}\,\mathrm{d}u\;p(z+u) \bigr)^{l} 
 \nonumber \\ &
 = 
 \frac{1}{2}\int_{\mathds{R}}\mathrm{d}z \bigr( \int_{0}^{1}\,\mathrm{d}u\;\hat p(z+u) \bigl)^{l} + o(\frac{1}{R^l}).
 \label{rear2}
\end{align}
Replacing these two formulas in (\ref{limFRP}) we have the representation
\begin{align}
 \mathcal{W}_{\rm int}&[p(\cdot)] 
 \nonumber
 \\ &
= \frac{\epsilon}{2l}\int_{\mathds{R}}\,\mathrm{d}z \hat p(z)^{l}-\frac{\epsilon}{2l}\int_{\mathds{R}}\,\mathrm{d}z
\bigl( \int_{0}^{1}\,\mathrm{d}u\; \hat p(z+u) \bigr)^{l}.
\label{limFRPhat}
\end{align}

Now consider $\hat p^*(\cdot)$, the symmetric decreasing rearrangement of $\hat p(\cdot)$.
The first term in (\ref{limFRPhat}) is invariant under
rearrangement. It remains to prove that the second term in (\ref{limFRPhat}) increases
upon rearrangement. We express it as follows (dropping $\epsilon/2l$):
\begin{align*}
&\int\limits_{\mathds{R}}\mathrm{d}z\;\Big( \int\limits_{0}^{1}\mathrm{d}u\;\hat{p}(z+u) \Big)^{l}\\
&=\int\limits_{\mathds{R}}\mathrm{d}z\int\limits_{\mathds{R}^l}\prod_{i=1}^{l}\mathrm{d}u_{i}\;\hat{p}(z+u_{i})\mathbb{1}_{[0,1]}(u_{i})\\
&\overset{(b)}=\int\limits_{\mathds{R}}\mathrm{d}z'\int\limits_{\mathds{R}^l}\prod_{i=1}^{l}\mathrm{d}u'_{i}\;\hat{p}(z'+R+u'_{i}+\frac{1}{2})\,\mathbb{1}_{[-\frac{1}{2},\frac{1}{2}]}(u'_{i})\\
&\overset{(c)}\leq \int\limits_{\mathds{R}}\mathrm{d}z'\int\limits_{\mathds{R}^l}\prod_{i=1}^{l}\mathrm{d}u'_{i}\;\hat{p}^*(z'+R+u'_{i}+\frac{1}{2})\,\mathbb{1}_{[-\frac{1}{2},\frac{1}{2}]}(u'_{i})\\
&\overset{(d)}= \int\limits_{\mathds{R}}\mathrm{d}z\int\limits_{\mathds{R}^l}\prod_{i=1}^{l}\mathrm{d}u_{i}\;\hat{p}^*(z+u_{i})\,\mathbb{1}_{[0,1]}(u_{i})\\
&= \int\limits_{\mathds{R}}\mathrm{d}z\;\Big( \int\limits_{0}^{1}\mathrm{d}u\; \hat{p}^*(z+u) \Big)^{l},
\end{align*}
where the equality in (b) is due to the changes of variables $z'=z-R$ and $u'_{i}=u_{i}-\frac{1}{2};\;i=1\dots l$, the 
inequality in (c) is due
Theorem \ref{BrascampsTheorem}, and the equality in (d) is obtained by 
first remarking that the indicator function $\mathbb{1}_{[-\frac{1}{2},\frac{1}{2}]}(u'_{i})$ is 
unchanged upon rearrangement and then by making the reverse changes 
of variables $z=z'+R$ and $u_{i}=u'_{i}+\frac{1}{2};\;i=1\dots l$.
So far we have obtained 
 \begin{align*}
 & \mathcal{W}_{\rm int}[p(\cdot)] 
 \nonumber
 \\ &
\geq \frac{\epsilon}{2l}\int\limits_{\mathds{R}}\,\mathrm{d}z \hat{p}^*(z)^{l}-\frac{\epsilon}{2l}\int\limits_{\mathds{R}}\,\mathrm{d}z
\Big( \int\limits_{0}^{1}\,\mathrm{d}u\; \hat{p}^*(z+u) \Big)^{l}.
\end{align*}
To obtain 
$\mathcal{W}_{\rm
int}[p(\cdot)] \geq \mathcal{W}_{\rm int}[p^*(\cdot)]$ it remains to reverse the steps \eqref{limFRP}-\eqref{limFRPhat}.  
\end{proof}

\section{Strict Monotonicity of Profile} \label{strIncrSection}

We establish Lemma \ref{strIncrProp} as a corollary of the following lemmas.
\begin{lemma}\label{strIncrLemma1}
If $p(\cdot)$ minimizes $\mathcal{W}[p(\cdot)]$, then it satisfies the DE equation.
\end{lemma}
\begin{proof}
Consider a profile  $p(\cdot)$ and a function $\nu(\cdot)$ such 
that $\lim\limits_{z\rightarrow \pm\infty}\nu(z)=0$. We compute the directional derivative of the potential $\mathcal{W}[p(\cdot)]$ in the direction of $\nu(\cdot)$,
\begin{equation*}
\mathrm{d}\mathcal{W}[p(\cdot)][\nu]=\lim\limits_{\delta\rightarrow 0}\frac{\mathcal{W}[p(\cdot)+\delta\nu(\cdot)]-\mathcal{W}[p(\cdot)]}{\delta}.
\end{equation*}
A calculation gives
\begin{align*}
\mathrm{d}\mathcal{W}[p(\cdot)][\nu]=\int\limits_{\mathds{R}}&\mathrm{d}z\,\nu(z)\Big\{1-(1-p(z))^{\frac{1}{r-1}}\\
&-\epsilon\int_0^1\mathrm{d}v\Big(\int_0^1\mathrm{d}u\,p(z+u-v) \Big)^{l-1} \Big\}.
\end{align*}
Now consider the function 
\begin{align*}
\nu_p(z)=& -\Big\{1-(1-p(z))^{\frac{1}{r-1}}
\nonumber \\ &
-\epsilon\int_0^1\mathrm{d}v\Big(\int_0^1\mathrm{d}u\,p(z+u-v) \Big)^{l-1} \Big\}.
\end{align*}
The directional derivative of $\mathcal{W}[p(\cdot)]$ in the direction of $\nu_p(\cdot)$ satisfies
\begin{equation}\label{eqnDirDerivNeg}
\mathcal{W}[p(\cdot)][\nu_p]\leq 0
\end{equation}
because the integrand is a square. 
Now assume that $p(\cdot)$ is a minimizing profile.  
In the case where equality is met in \eqref{eqnDirDerivNeg}, $p(\cdot)$ satisfies the DE equation.
Consider the case where the inequality is strict. 
Then,
\begin{equation*}
\mathrm{d}\mathcal{W}[p(\cdot)][\nu_p]=\lim\limits_{\delta\rightarrow 0}\frac{\mathcal{W}[p(\cdot)+\delta\nu_p(\cdot)]-\mathcal{W}[p(\cdot)]}{\delta}<0,
\end{equation*}
and we can find $\delta_0$ small enough such that 
\begin{equation*}
\mathcal{W}[p(\cdot)+\delta_0\nu_p(\cdot)]<\mathcal{W}[p(\cdot)].
\end{equation*}
So $p(\cdot)$ cannot be a minimizing profile, and this concludes the proof by contradiction.
\end{proof}

\begin{lemma}\label{strIncrLemma3}
If $p(\cdot)$ is increasing and satisfies the DE equation, then it cannot be strictly flat
on an interval $]a, b[ \subset \mathds{R}$.
\end{lemma}
\begin{proof}
If $p(\cdot)$ satisfies the DE equation, then
\begin{align*}
1-(1-p(z))^{\frac{1}{r-1}}=\epsilon\int_0^1\mathrm{d}v\Big(\int_0^1\mathrm{d}u\, p(z+u-v) \Big)^{l-1}
\end{align*}
By taking the derivative on each side, we find that 
\begin{align}
(1-p(z)&)^{\frac{1}{r-1}-1}\frac{1}{r-1}p'(z)=\epsilon\int_0^1\mathrm{d}v\,(l-1)\times\nonumber\\
\Big(\int_0^1\mathrm{d}u&\,p(z+u-v) \Big)^{l-2}\int_0^1\mathrm{d}w\,p'(z+w-v).\label{eqnDerivBothSides}
\end{align}
Notice that $\int_0^1\mathrm{d}w\,p'(z+w-v)=\int_0^1\mathrm{d}w\,\frac{\mathrm{d}}{\mathrm{d}w}p(z+w-v)=p(z+1-v)-p(z-v)$.

Now assume that there exists a flat spot of $p(\cdot)$ for $z\in]a,b[$ where it takes some value $p_{\rm\tiny flat}$. 
We consider ``maximal'' intervals $]a, b[$ such that $p(\cdot)$ takes 
values different than $p_{\rm\tiny flat}$ for all $z\not\in]a,b[$. On this flat spot, \eqref{eqnDerivBothSides} becomes
{\small
\begin{align}\label{derivzero}
0=\int_0^1\mathrm{d}v\,\Big(\int_0^1\mathrm{d}u&\,p(z+u-v) \Big)^{l-2}(p(z+1-v)-p(z-v)).
\end{align}
}
We will now show that this equality cannot be satisfied.

Let us first consider the case when $a$ and $b$ are finite. Since $]a, b[$ is maximal and $p(\cdot)$ is increasing, 
we know that $0<p_{\rm\tiny flat}<p_{\rm\tiny MAP}$, $p(z)<p_{\rm\tiny flat}$ for all $z<a$ and $p(z)>p_{\rm\tiny flat}$ for all $z>b$. Now 
let us fix $z\in [b-1+\delta_0,b]$, where $0<\delta_0<1$. For such a $z$, the 
equality \eqref{derivzero} holds. But for such $z$ and for all $0 \leq v <1$, 
$\int_{0}^1\mathrm{d}u\,p(z+u-v)\geq \int_{u>v}\mathrm{d}u\,p(z+u-v)\geq (1-v)p_{\tiny\rm flat} > 0$. Thus we should have  
$p(z+1-v)=p(z-v)$ for a.e $v\in[0,1]$. This is not possible. Indeed, 
take $v\in [0, \delta_1]$ with $0<\delta_1<\delta_0$ is small enough so that $z-v<b<z+1-v$
and thus $p(z+1-v) - p(z-v) > p(b + \delta_0 - \delta_1) - p(b) > 0$. 
These arguments prove that an increasing solution of DE cannot be flat for $z\in [b-1+\delta_0,b]$. 
We repeat the argument on $[b-k+\delta_0,b-(k-1)+\delta_0]$ for all $1<k<K$ such that $b-K+\delta_0<a$ and find 
that an increasing solution of DE cannot be flat on each of those intervals, 
and thus on $[b-(K-1)+\delta_0,b]$. Finally, we repeat the argument on the last interval $[a,b-(K-1)+\delta_0]$ and 
deduce that an increasing solution of DE cannot be flat on $[a,b]$.

Next, we consider the case when $a= -\infty$. In this case, we have 
that $p(z)=0$ for all $z\leq b$ and $p(z)>0$ for all $z>b$. The analysis is similar than the preceding one. First fix $z$ in the 
interval $\in [b-1 + \delta_0, b]$, $0<\delta_0<1$.  Equation \eqref{derivzero} is satisfied for such $z$. Now take
$v\in [0, \delta_1]$ with $0< 2\delta_1 < \delta_0$. Then $p(z+1-v) -p(z-v) = p(z+1-v) > p(b+\delta_0-\delta_1)>0$,  
and for $1-\delta_0 +2\delta_1 < u < 1$ we have $p(z+u-v) > p(b+\delta_1)$ so
$\int_0^1 du\, p(z+u -v) >  (1+\delta_0 - 2\delta_1) p(b +\delta_1) >0$. Thus the right hand side of \eqref{derivzero} does not vanish which is a contradiction. 
We carry out the same analysis 
as above for $z\in [b-k+\delta_0,b-(k-1)+\delta_0]$ for $k\in\mathbb{N}$, and thus deduce that $p(\cdot)$ cannot be flat on $[-\infty,b]$.

Finally, consider the case when $b= +\infty$. The analysis is essentially symmetric to the preceding one.
In this case we have $p(z)=p_{\rm\tiny MAP}$ for $z>a$ and 
$p(z)<p_{\rm\tiny MAP}$ for $z<a$. First fix $z$ in the interval $[a, a+1-\delta_0]$. For such $z$ \eqref{derivzero} holds.
For $v\in [1-\delta_1, 1]$ with $\delta_1 < \delta_0$ we have 
$p(z+1-v) - p(z-v) > p_{\rm \tiny MAP} - p(a - \delta_0  +\delta_1) >0$. Moreover it is clear that
$\int_0^1 du\, p(z+u -v) > 0$. So the right hand side of \eqref{derivzero} cannot vanish, and we arrive at a contradiction.
We repeat the argument for  $z\in [a+k-\delta_0 , a + (k+1) -\delta_0]$ for $k\in\mathbb{N}$, and conclude that $p(\cdot)$ cannot be flat on $[a, +\infty]$.

%
%
%
\end{proof}

\begin{proof}[Proof of Lemma \ref{strIncrProp}]
The proposition follows directly from Lemmas \ref{strIncrLemma1} 
and \ref{strIncrLemma3}.
\end{proof}

We now discuss a lemma that is interesting in itself but not necessary for our results.
\begin{lemma}\label{strIncrLemma2}
If $p(\cdot)$ minimizes $\mathcal{W}[p(\cdot)]$ then it cannot have a flat spot, i.e., $p(z)=p_{\tiny\rm flat}$ with $0<p_{\rm \tiny flat} <p_{\rm \tiny MAP}$
 in a bounded interval $z\in [a,b]$ such that $b-a>1$.
\end{lemma}
\begin{proof}
Suppose that $p(\cdot)$ is an increasing minimizing profile and that it has a constant value $0<p_{\rm \tiny flat} <p_{\rm \tiny MAP}$ on a bounded interval of length 
greater than $1$. We will construct another profile that has strictly less energy. 

We start by expressing the single potential as follows
\begin{align*}
\mathcal{W}_{\rm single}[p(\cdot)]&=\int_{-\infty}^{a}\mathrm{d}z\,W_{\rm s}(p(z))+\int_{a}^{b-1}\mathrm{d}z\,W_{\rm s}(p(z))\\&\qquad\qquad +\int_{b-1}^{+\infty}\mathrm{d}z\,W_{\rm s}(p(z))\\
&=\int_{-\infty}^{a}\mathrm{d}z\,W_{\rm s}(p(z))+(b-1-a)W_{\rm s}(p(a))\\&\qquad\qquad +\int_{b-1}^{+\infty}\mathrm{d}z\,W_{\rm s}(p(z)).
\end{align*}
By applying the change of variables $z^\prime=z-(b-1)+a$ on the rightmost integral, we express it as
\begin{align*}
\int_{b-1}^{+\infty}\mathrm{d}z\,W_{\rm s}(p(z))=\int_{a}^{+\infty}\mathrm{d}z\,W_{\rm s}(p(z+b-1-a)).
\end{align*}
We define the profile $\tilde{p}$ by
\begin{equation*}
\tilde{p}(z)=
\begin{cases}
p(z) & \text{if}\quad z\leq a,
\\
p(z+b-1-a) & \text{if}\quad z>a,
\end{cases}
\end{equation*}
We thus obtain
\begin{align*}
\mathcal{W}_{\rm single}[p(\cdot)]=\mathcal{W}_{\rm single}[\tilde{p}(\cdot)]+(b-1-a)W_{\rm s}(p_{\rm \tiny flat}).
\end{align*}
Note that since $0<p_{\rm \tiny flat} <p_{\rm \tiny MAP}$ we have $W_{\rm s}(p_{\rm \tiny flat}) >0$. 
Thus $\mathcal{W}_{\rm single}[p(\cdot)] > \mathcal{W}_{\rm single}[\tilde{p}(\cdot)]$.

For the interaction potential, we prove that $\mathcal{W}_{\rm int}[p(\cdot)]=\mathcal{W}_{\rm int}[\tilde{p}(\cdot)]$. Indeed,
\begin{align}
\mathcal{W}_{\rm int}[p(\cdot)]&=\frac{\epsilon}{l}\int_{\mathds{R}}\mathrm{d}z\Big\{ p(z)^l-\Big(\int_0^1\mathrm{d}u \,p(z+u) \Big)^l\Big \}\nonumber\\
=&\frac{\epsilon}{l}\int_{-\infty}^{a}\mathrm{d}z\Big\{ p(z)^l-\Big(\int_0^1\mathrm{d}u \,p(z+u) \Big)^l\Big \}\label{eqnTerm1}\\
+&\frac{\epsilon}{l}\int_{a}^{b-1}\mathrm{d}z\Big\{ p(z)^l-\Big(\int_0^1\mathrm{d}u \,p(z+u) \Big)^l\Big \}\label{eqnTerm2}\\
+&\frac{\epsilon}{l}\int_{b-1}^{+\infty}\mathrm{d}z\Big\{ p(z)^l-\Big(\int_0^1\mathrm{d}u \,p(z+u) \Big)^l\Big \}\label{eqnTerm3}.
\end{align}
We use the same definition of $\tilde{p}(\cdot)$ as above, and denote the 
functionals in \eqref{eqnTerm1}, \eqref{eqnTerm2}, and \eqref{eqnTerm3} 
by $\mathcal{T}_1[p(\cdot)]$, $\mathcal{T}_2[p(\cdot)]$, and $\mathcal{T}_3[p(\cdot)]$ respectively. Observe that
\begin{itemize}
\item $\mathcal{T}_1[p(\cdot)]=\mathcal{T}_1[\tilde{p}(\cdot)]$ since
\begin{align*}
p(z) &= \tilde{p}(z) \;\quad\text{if}\quad z<a \\
p(z+u) &= \tilde{p}(z +u) \;\quad\text{if}\quad z<a, \,\,\text{and}\,\,0<u<1.
\end{align*}
\item $\mathcal{T}_2[p(\cdot)]=0$ since $p(z)=p(z+u)=p_{\tiny\rm flat}$ when $a<z<b-1$.
\item $\mathcal{T}_3[p(\cdot)]=\mathcal{T}_2[\tilde{p}(\cdot)]+\mathcal{T}_3[\tilde{p}(\cdot)]$ since, by the change of variables $z^\prime=z-(b-1)+a$, we have
\begin{align*}
\mathcal{T}_3[p(\cdot)]&=\frac{\epsilon}{l}\int\limits_a^{+\infty}\mathrm{d}z\Big\{p(z+b-1-a)^l\\
&\qquad-\Big(\int_0^1 \mathrm{d}u\,p(z+b-1-a+u) \Big)^l \Big\}\\
&=\frac{\epsilon}{l}\int\limits_a^{+\infty}\mathrm{d}z\Big\{\tilde{p}(z)^l-\Big(\int_0^1 \mathrm{d}u\,\tilde{p}(z+u) \Big)^l \Big\}.
\end{align*}
\end{itemize}
Thus $\mathcal{T}_1[p(\cdot)] + \mathcal{T}_2[p(\cdot)] + \mathcal{T}_3[p(\cdot)] = \mathcal{T}_1[\tilde p(\cdot)] + \mathcal{T}_2[\tilde p(\cdot)] 
+ \mathcal{T}_3[\tilde p(\cdot)]$

Combining these results we get 
\begin{equation*}
 \mathcal{W}[p(\cdot)]=\mathcal{W}[\tilde{p}(\cdot)]+(b-1-a)\,W_{\rm s}(p_{\tiny\rm flat})>\mathcal{W}[\tilde{p}(\cdot)]
\end{equation*}
\end{proof}

\section{Tightness of the minimizing sequence} \label{tightnessSection}

\begin{proof}[Proof of Lemma \ref{tightnesslemma}]
Consider a minimizing sequence of cdfs $p_n(\cdot)$, i.e., satisfying \eqref{minSequence}. Fix any $\delta>0$ and suppose that 
\begin{equation}\label{nottight}
 p_n(M) - p_n(-M) < (1-\delta)p_{\text{\tiny MAP}}.
\end{equation}
We will show that \eqref{nottight} implies that necessarily $M \leq c/\delta^2$ for a fixed constant $c>0$. Taking the 
contrapositive we find that: choosing 
$M_\delta = c^\prime / \delta^2$ with $c^\prime >c$ implies that any minimizing sequence satisfies 
$p_n(M_\delta) - p_n(-M_\delta) > (1-\delta)p_{\text{\tiny MAP}}$.

From Lemma \ref{poslem} we have
\small
\begin{equation}\label{first}
 \mathcal{W}[p_n(\cdot)]\geq \mathcal{W}_{\text{\small single}}[p_n(\cdot)]-\frac{\epsilon p_{\text{\tiny MAP}}^l}{2l} \geq 
 \int_{-M}^M dz\, {W}_{s}(p_n(z))-\frac{\epsilon p_{\text{\tiny MAP}}^l}{2l}.
\end{equation}
\normalsize
Now, assuming \eqref{nottight} there must be a mass at least $\delta p_{\text{\tiny MAP}}$ outside of the interval
$[-M, M]$. Thus we have $p_n(-M) \geq \delta p_{\text{\tiny MAP}}/2$ or $p_{\text{\tiny MAP}} - p_n(M) \geq \delta p_{\text{\tiny MAP}}/2$. 
Recall that $p_n\in\mathcal{S}^{\prime\prime}_0$ so $p_n(0) = p_{\rm\tiny MAP}/2$. Therefore in $[-M, 0]$ or in $[0, M]$
 the profile $p_n(z)$ must be $\delta p_{\text{\tiny MAP}}/2$ away from the minima $0$ and $p_{\text{\tiny MAP}}$ of ${W}_{s}$.
Moreover, one can check that
 $W_s(p)$ has a parabolic 
shape near the minima at $0$ and $p_{\text{\tiny MAP}}$ so that away from these minima $W_s(p) \geq C\delta^2 p_{\text{\tiny MAP}}^2/4$ 
for a constant $C>0$ depending only on $l$.
These remarks imply
\begin{align}\label{second}
\int_{-M}^M dz\, {W}_{s}(p_n(z)) \geq \frac{1}{4}M C\delta^2 p_{\text{\tiny MAP}}^2.
\end{align}
Since $p_n(\cdot)$ is a minimizing sequence, for $n$ large enough its cost must be smaller than the cost of a fixed reference profile, say
$\rho(z) = 0, z \leq 0$, $\rho(z) =p_{\text{\tiny MAP}}, z > 0$. More formally,
\begin{equation}\label{third}
 \mathcal{W}[p_n(\cdot)] < \mathcal{W}[\rho(\cdot)] = - \frac{\epsilon}{l(l+1)} p_{\text{\tiny MAP}}^l\,.
\end{equation}
Finally, combining \eqref{first}, \eqref{second} and \eqref{third} we find that 
\begin{equation}
 M \leq 2\frac{\epsilon(l-1)}{l(l+1)}\frac{p_{\text{\tiny MAP}}^{l-2}}{C\delta^2}.
\end{equation}
\end{proof}

\section{Explicit Expressions of Kernel Function}\label{explicitkernel}

In this section, we compute the kernel function $V$ of section \ref{subproblem} and illustrate some of its properties.
In particular we show that it is  strictly convex in a set of positive measure.

Recall that the function is totally symmetric under permutations. It is therefore enough to compute it in a fixed sector 
$S_{x}=\{\mathbf{x}=(x_1,\cdots,x_l):\,x_i\geq x_j \;\mathrm{if}\;i<j \}$. We express $V$ in terms of 
the distances $d_{1i}$, which are ordered such that $d_{1i}<d_{1j}\,\mathrm{if}\,i<j$.

Let us first discuss the explicit examples $l=2$ and $l=3$. For $l=2$ an explicit computation yields, 

\begin{equation*}
V_{(l=2)}(d_{12})=
\begin{cases}
-\frac{1}{2} & \text{if}\quad d_{12}\geq 1,
\\
-\frac{1}{2}+\frac{1}{6}(1-d_{12})^3 & \text{if}\quad d_{12}<1.
\end{cases}
\end{equation*}
By taking the second derivative it is easy to see that $V_{(l=2)}$ is convex everywhere, and strictly convex for $d_{12}<1$.

For $l=3$, we have $d_{12} < d_{13}$ and the computation yields, 
\begin{equation*}
V_{(l=3)}(d_{12},d_{13})=
\begin{cases}
V_{(l=2)}(d_{12}) & \text{if}\,\, d_{13}\geq 1,
\\
\\
-\frac{1}{2}+\frac{1}{6}(1-d_{12})^{3}\\
+\frac{1}{12}(1-d_{13})^{4}& \text{if}\,\,d_{13}<1.\\
+\frac{1}{6}d_{12}(1-d_{13})^{3}
\end{cases}
\end{equation*}
For $d_{13}<1$ the Hessian is
\begin{align*}
\left(\begin{array}{ccc}
\displaystyle 1-d_{12} & \displaystyle -\frac{1}{2}(d_{13}-1)^{2} \\
\displaystyle -\frac{1}{2}(d_{13}-1)^{2} & \displaystyle -(d_{13}-1)(1+d_{12}-d_{13})\\
\end{array}\right)
\end{align*}
and the corresponding eigenvalues are:
\begin{align*}
\lambda_{1,2}=\frac{1}{2}\{2&-2d_{13}-d_{12}d_{13}+d_{13}^{2}\pm\sqrt{\Delta}\},
\end{align*}
where 
\begin{align*}
\Delta=1&+4d_{12}^{2}-4d_{13}-8d_{12}d_{13}-4d_{12}^{2}d_{13}+10d_{13}^{2}\\
&+8d_{12}d_{13}^{2}+d_{12}^{2}d_{13}^{2}-8d_{13}^{3}-2d_{12}d_{13}^{3}+2d_{13}^{4}.
\end{align*}
A plot of the eigenvalues  shows that they are non-negative in the region $0\leq d_{12}\leq d_{13}\leq 1$. 
In fact, one eigenvalue is strictly positive everywhere in this region, and the other is strictly positive 
everywhere in this region except at the boundary $d_{13}=1$, where it becomes equal to zero. 
This is consistent with the fact that $V_{(l=3)}(d_{12},d_{13})=V_{(l=2)}(d_{12})$ when $d_{13}\geq 1$. For $d_{13}\geq 1$, the Hessian always has 
a vanishing eigenvalue, and a 
strictly positive one when $d_{12} <1$. For $d_{12} \geq 1$, the kernel $V_{(l=3)}(d_{12},d_{13})$ is constant and both eigenvalues vanish.
To summarize the kernel is always convex, and strictly convex for $0\leq d_{13} < 1$.

These results can be generalized for all $l$. We find the general expression of the kernel 
\begin{align*}
V_{l}(d_{12},\cdots,d_{1l})=\sum\limits_{k=2}^{l}\sum\limits_{m=k}^{l}&\frac{(1-d_{1m})^{m-k+3}}{(m-k+3)(m-k+2)}\\
&\times\Big(\sum\limits_{\mathcal{S}\subseteq\{2,\ldots,m-1\}}\prod\limits_{n \in \mathcal{S}}d_{1n}\Big).
\end{align*}
The corresponding Hessian $(H_{ij})$ is a symmetric matrix of dimension $(l-1)\times(l-1)$ with matrix elements that are 
polynomials in $d_{1i}$, $i=1,\dots, l$. In particular, at $d_{1i}=0$ for all $i$ we have 
$H_{ii}=1$ and $H_{ij}=-\frac{1}{j+1}+\sum_{m=j+1}^{l}\frac{1}{(m-1)(m-2)}=-\frac{1}{l-1};\,j>i$. 
Defining $\mathbf{v}$ as the $(l-1)$-dimensional vector of $1$'s and denoting by $\mathbb{1}$ the $(l-1)$-dimensional identity matrix, 
we remark that $H$ at the origin can be expressed as
\begin{equation*}
H=(1+\frac{1}{l-1})\mathbb{1}-\frac{1}{l-1}\mathbf{v}\mathbf{v}^\mathrm{T}.
\end{equation*}
The eigenvalues of this matrix are $1+\frac{1}{l-1}$ and $1$ (with $1+\frac{1}{l-1}$ having degeneracy $l-2$).
Since these eigenvalues are strictly positive, the Hessian is strictly positive definite at the origin, and thus (by continuity)
also in a small neighborhood of the origin. Thus $V$ is a strictly convex function of $d_{1i}$, $i=2,\dots, l$ in a small neighborhood of the origin.

\noindent{\bf Acknowledgment.} We thank Marc Vuffray for discussions
on rearrangement inequalities and for bringing reference \cite{Brascamp}
to our attention. We in particular thank Tom Richardson for pointing
out the subtle way in which uniqueness can fail in a system despite
the strict displacement convexity of its underlying functional.

\end{document}